\documentclass[conference,letterpaper]{IEEEtran}

\usepackage[utf8]{inputenc} 
\usepackage[T1]{fontenc}
\usepackage{url}
\usepackage{ifthen}
\usepackage{cite}
\usepackage[cmex10]{amsmath} 

\usepackage{cite}
\usepackage{amsmath}
\usepackage{amsfonts,amssymb,amsthm, bbm,braket}
\usepackage{algorithmic}
\usepackage{graphicx}
\usepackage{textcomp}
\usepackage{xcolor}
\usepackage[ruled, linesnumbered, vlined]{algorithm2e}
\usepackage{hyperref}
\usepackage{tikz}

\newtheorem{definition}{Definition}

\newtheorem{theorem}{Theorem}
\newtheorem{lemma}{Lemma}

\newcommand\submittedtext{%
  \footnotesize This work has been submitted to the IEEE for possible publication. Copyright may be transferred without notice, after which this version may no longer be accessible.}

\newcommand\submittednotice{%
\begin{tikzpicture}[remember picture,overlay]
\node[anchor=south,yshift=10pt] at (current page.south) {\fbox{\parbox{\dimexpr0.65\textwidth-\fboxsep-\fboxrule\relax}{\submittedtext}}};
\end{tikzpicture}%
}

\begin{document}
\title{Deep Unfolding of Fixed-Point Based Algorithm for Weighted Sum Rate Maximization} 


\author{%
  \IEEEauthorblockN{Jan Christian Hauffen}
  \IEEEauthorblockA{Communications and Information Theory\\
                    Technical University of Berlin\\
                    Berlin\\
                    Email: j.hauffen@tu-berlin.de}
  \and
  \IEEEauthorblockN{Chee Wei Tan}
  \IEEEauthorblockA{Nanyang Technological University\\ 
                    Singapore\\
                    Email: cheewei.tan@ntu.edu.sg}
  \and
  \IEEEauthorblockN{Giuseppe Caire}
  \IEEEauthorblockA{Communications and Information Theory\\
                    Technical University of Berlin\\
                    Berlin\\
                    Email: caire@tu-berlin.de}
}

\maketitle

\submittednotice
\begin{abstract}
In this paper, we propose a novel approach that harnesses the standard interference function, specifically tailored to address the unique challenges of non-convex optimization in wireless networks. We begin by establishing theoretical guarantees for our method under the assumption that the interference function exhibits log-concavity. Building on this foundation, we develop a Primal-Dual Algorithm (PDA) to approximate the solution to the Weighted Sum Rate (WSR) maximization problem. To further enhance computational efficiency, we leverage the deep unfolding technique, significantly reducing the complexity of the proposed algorithm. Through numerical experiments, we demonstrate the competitiveness of our method compared to the state-of-the-art fractional programming benchmark, commonly referred to as FPLinQ.
\end{abstract}

\section{Introduction}
With the advent of Open Radio Access Network (OpenRAN) architectures for 6G, which aim to enable flexible, interoperable, and software-defined solutions, there is a pressing need for scalable and efficient algorithms that can address the complex interference management challenges in these emerging networks. In this paper, we propose a novel approach that integrates the standard interference function framework with difference-of-convex programming to address the challenging problem of maximizing the Weighted Sum Rate (WSR) in a Gaussian interference channel. This problem, which involves optimizing power control strategies, is inherently non-convex due to the complex interplay of interference in wireless networks. Our approach aims to overcome these challenges while providing efficient solutions suitable for next-generation wireless systems.

The standard interference function framework has been instrumental in optimizing wireless networks, providing a mathematical basis for modeling and managing interference. With the advent of 6G networks—featuring ultra-dense deployments, higher frequencies, and greater connectivity demands—scalable algorithms for interference mitigation are increasingly vital. By modeling interference as a function of transmit powers, spatial configurations, and transmitter-receiver distances, this framework captures the trade-offs between interference and signal quality. Coupling it with advanced techniques like difference-of-convex programming enables the design of scalable resource allocation strategies to address the complex interference challenges in 6G networks.

Advances in the past decade have significantly expanded the applicability and effectiveness of the interference function framework in wireless network optimization \cite{TWC2012}. One notable development is the incorporation of sophisticated interference models that capture the effects of various factors such as channel fading, multi-path propagation, and spatial correlation. These refined models enable more accurate prediction and control of interference dynamics, leading to improved network performance and spectral efficiency.

As the optimization of non-convex wireless networks becomes a central focus in research and deployment, the need for innovative techniques to address its inherent complexities has grown substantially. Traditional optimization methods often fall short in managing the intricacies of non-convexity and meeting the stringent performance demands of modern wireless systems. In this paper, we propose an optimization framework that leverages the standard interference function, tailored to the unique challenges of non-convex optimization in wireless networks. To enhance scalability and adaptability, we integrate deep learning into the design of optimization algorithms for 6G protocols. Deep learning facilitates data-driven modeling and prediction of interference dynamics in ultra-dense, heterogeneous network environments, where conventional methods are inadequate. By incorporating deep unfolding techniques, our approach seamlessly combines model-based optimization with data-driven learning, resulting in structured, interpretable, and scalable algorithms. Through rigorous theoretical analysis and empirical validation, we demonstrate that our framework delivers robust and efficient solutions to the high-dimensional and non-convex resource allocation challenges of next-generation 6G wireless networks.

Optimization algorithms can be unfolded into a finite number of iterations, treating these stages as layers of a deep neural network, with parameters trainable via stochastic gradient methods. This approach, known as deep unfolding \cite{gregor2010learning}, has recently gained attention in  wireless communications, as demonstrated by the algorithm-unfolding-based network proposed in \cite{li_graph-based_2022}.

Difference-of-convex programming (DCP) is a framework of optimizing non-convex functions, when the given objective can be expressed as the difference of two convex functions. Specifically, a simple algorithm is provided, where a global estimator is constructed and a series of so defined convex sub-problems are solved, this method is usually referred to as the difference-of-convex function algorithm (DCA). Furthermore, this optimization framework can be connected to a range of different opitmization frameworks, e.g. the Succesive Convex Approximation, the Majorization-Minimization \cite{sun2016majorization}, Expectation-Maximization, or the Concave–Convex procedure (CCCP) \cite{yuille2001concave, lanckriet2009convergence, boydccp1,boydccp2}. An overview of the connection to other optimization frameworks is given in \cite{le2018dc}. DCP was first mentioned in \cite{tao1986algorithms}. Moreover, in \cite{yao2023globally} the convergence properties to a global optima are investigated as well as a Bregman distance based framework for given convex-sub-problems is developed.

DCP has already been investigated in the optimization of wireless networks, e.g. \cite{vucic2010dc, al2011achieving}. In \cite{vucic2010dc} the approach was applied directly to the maximization of the WSR. Furthermore, the classical difference-of-convex algorithm was investigated and convergence could be proven if the respective interference function is concave. In \cite{al2011achieving} DCP was applied to find a global optimum of the problem of sum rate maximization with a total power constraint.

In this paper, we propose a general optimization approach for WSR maximization, designed to handle $\log$-concave interference functions. We demonstrate that an already known fixed-point algorithm can be generalized as a standard power control algorithm when the underlying interference function exhibits log-concavity. This result is then used to develop a Primal Dual Algorithm to approximate a solution of the WSR problem. The complexity of the Primal Dual Algorithm is then reduced by employing the deep unfolding approach. Numerical experiments show that the proposed learned algorithm is comparable to the state-of-the-art benchmark based on Fractional Programming, \cite{shen_fplinq_2017, shen_fractional_2018, shen_fractional_2018-1}. 

\noindent\textbf{Notation:} We denote matrices, respectively vectors with bold capital or bold small letters, i.e. $\mathbf{A, a} $. $\mathbb{E}[x]$ is the expected value of a random variable $x$.

\section{Mathematical Background and Problem Formulation}
\subsection{Standard Interference Function Framework}
In \cite{yates1995framework} Yates introduces a framework for power control and introduces a comprehensive framework for managing transmit power levels from mobile devices in order to optimize the overall performance of the cellular network. We adapt the definition of a standard interference function from there.
\begin{definition}
A function $\mathcal{I}\,:\,\mathbb{R}^K_+ \mapsto \mathbb{R}_+$ is called \textbf{Standard Interference Function} \cite{yates1995framework}, $\mathcal{I}\,:\,\mathbb{R}^K_+ \mapsto \mathbb{R}_+$, if
    \begin{align}
        \text{Positivity: }&\mathcal{I}({\bf p}) > 0, \, \, {\bf p} > 0\\
        \text{Scalability: }&\alpha \mathcal{I} ({\bf p}) > \mathcal{I}(\alpha{\bf p}), \, \, \alpha > 1 \\
        \text{Monotonicity: }& \mathcal{I}({\bf p}) \geq \mathcal{I}({\bf p}'), \, \,  {\bf p} \geq {\bf p}'
    \end{align}
    Moreover, let $\mathcal{I}\,:\,\mathbb{R}^K_+ \mapsto \mathbb{R}_+$ be a Standard Interference Function. We call $\mathcal{I}$ \textbf{feasible} if there exists some $\bf p\geq 0$ s.t. 
    \begin{align}
        \bf p \geq \mathcal{I}(\bf p).
    \end{align}
\end{definition}
Moreover it is known, if ${\bf \mathcal{I}}(\bf p)$ is standard and feasible, then the so called standard power control algorithm 
\begin{align*}
    {\bf p}^{(k+1)} ={\bf \mathcal{I}}({\bf p}^{(k)}) 
\end{align*}
converges to a fixed point, see Theorem 2 in \cite{yates1995framework}. 

\subsection{Difference-of-convex Programming}
In DCP one considers problems of the form
\begin{align}
    \min_{{\bf x}\in\Omega} f({\bf x}) - g({\bf x})\label{eq:DCP}
\end{align}
where $f,g$ are convex and smooth functions.

To approximate the solution of \eqref{eq:DCP} one can then iteratively solve the following convex-sub-problems, these are given by approximating $g(\cdot)$ at ${\bf x}^{(k)}$ with a first-order Taylor expansion,
\begin{align}
    {\bf x}^{(k + 1)}\in\arg\min_{{\bf x}\in\Omega} f({\bf x})-\left(g({\bf x}^{(k)}) + \nabla g({\bf x}^{(k)})^T({\bf x} - {\bf x}^{k})\right).\label{eq:orig_DCA}
\end{align}
Iterations \eqref{eq:orig_DCA} are usually referred to as the difference-of-convex algorithm (DCA).

\subsection{System Model and Problem Formulation}
In the following, a Gaussian interference channel with $K$ transmitter-receiver pairs is considered. The signal at the receiver of link i is given by 
\begin{align}
    Y_i(t) = \sum_{j=1}^K h_{ij}\Tilde{X}_j(t)+Z_i(t)\, , i=1,\dots,K \label{eq:K_user_channel}
\end{align} 
with power constraint $\mathbb{E}[|\Tilde{X}_i(t)|^2]\leq P_i $, where $\Tilde{X}_j(t)$ is the transmitted symbol of transmitter $j$, $Y_i(t)$ the received signal at receiver $i$ and $Z_i(t)\thicksim\mathcal{CN}\left(0,\sigma^2\right)$ for each discrete time point $t$.

Let ${\bf G}\in\mathbb{R}^{K\times K}$ be the system matrix containing the channel gain of transmitter $j$ to receiver $i$. The WSR problem formulation, for strict positive weights $w_i$, is usually stated as follows
\begin{align}\label{eq:sr}
\max_{\bf p}\sum_i^K & w_i\log\left(1 + \frac{\mathbf{G}_{ii}{\bf p}_i}{\mathcal{I}_i({\bf p})}\right).
\end{align}

For practical reasons in the great majority of works it is assumed that interference is treated as noise, e.g. \cite{shen_fplinq_2017, shen_fractional_2018, shen_fractional_2018-1, naderializadeh_itlinq_2014, wu_flashlinq_nodate,geng_optimality_2015, yi_optimality_2015, cui_spatial_2019, lee_graph_2020, shelim_geometric_2022, lee2018deep}, which corresponds to the affine linear interference function
\begin{align}
\mathcal{I}_i({\bf p}) = \sum_{j\neq i} {\bf G}_{ij}{\bf p}_j + \sigma_i,\label{eq:aff_int}
\end{align} where $\sigma_i$ is the additive white Gaussian noise (AWGN) variance.

As the following theory can be applied to a wider class of interference functions, the framework of DCP is applied to problem \eqref{eq:sr} for an arbitrary standard, feasible and $\log$-concave interference function.

\section{WSR-Maximization for $\log$-concave Interference Functions}
In the following we discuss the difference-of-convex functions framework to approximate a solution to \eqref{eq:sr} for an arbitrary and $\log$-concave interference function. Especially for a special case of \eqref{eq:sr} we derive a closed form fixed-point algorithm, guaranteed to converge. 
\vspace{-1.5mm}
\subsection{Special Case}\label{sec:special_case}
Let us examine an approximation to \eqref{eq:sr} where a special class can be solved using fixed-point algorithms in polynomial time. Consider the following problem formulation:
\begin{subequations}\label{eq:WS_log}
\begin{align}
    \max_\mathbf{p} &\sum_i^K w_i\log\left(\frac{\mathbf{G}_{ii}{\bf p}_i}{\mathcal{I}_i({\bf p})}\right) \\
    \text{s.t. }& 0 < {\bf p}_i \leq P_\text{max}\, \forall i=1,\dots, K.
\end{align}
\end{subequations}
Using the properties of the logarithm we have
\vspace{-0.5mm} 
\begin{align*}
\sum_i^K w_i\log\left(\frac{\mathbf{G}_{ii}{\bf p}_i}{\mathcal{I}_i({\bf p})}\right) & = \sum_i^K w_i\log (\mathbf{G}_{ii}{\bf p }_i) \\&- \sum_i^K w_i \log\left(\mathcal{I}_i({\bf p })\right).
\end{align*}
If $\mathcal{I}_i({\bf p })$ is $\log$-concave the original problem is a difference-of-convex problem. The corresponding DCA is given by 
\begin{align*}
    {\bf p}^{(k+1)}&\in\arg\max_{\bf p}\sum_i^K w_i\log (\mathbf{G}_{ii}{\bf p }_i)\\& - \sum_i^K w_i \log\left(\mathcal{I}_i({\bf p }^{(k)})\right) \\& -\nabla_{{\bf p}} \left(\sum_i^K w_i\log\left(\mathcal{I}_i({\bf p}^{(k)})\right)\right)^T\left({\bf p} - {\bf p}^{(k)}\right).
\end{align*}
The first order optimality conditions for these convex sub-problem are given as
\begin{subequations}\label{eq:grads}
    \begin{align}
    \nabla_{\bf p} &\left(\sum_i^K w_i\log \left(\mathbf{G}_{ii}{\bf p }^{(k+1)}_i\right)\right) \\ & = \nabla_{\bf p} \left(\sum_i^K w_i\log\left(\mathcal{I}_i({\bf p}^{(k)})\right)\right).
\end{align}
\end{subequations}

This structure can also be obtained by the CCCP, see \cite{yuille2001concave}, which can also be interpreted as an instance of the DC framework. Moreover, if $\mathcal{I}_i({\bf p })$ is $\log$-concave, then \eqref{eq:grads} is guaranteed to converge to a stable point of \eqref{eq:WS_log}.

By isolating ${\bf p}^{(k + 1)}$ in \eqref{eq:grads}, one can obtain closed update rule
\begin{align}
    \frac{w_i}{\mathbf{p}^{(k+1)}_i} &= \left(\nabla_{\bf p} \left(\sum_j^K w_j\log\left(\mathcal{I}_j({\bf p}^{(k)})\right)\right)\right)_i \nonumber\\
    {\bf p}^{(k+1)}_i &=\min\left\lbrace w_i \left[\sum_j^K\frac{w_j\nabla_{\bf p} \mathcal{I}_j({\bf p}^{(k)})_i}{\mathcal{I}_j({\bf p}^{(k)})}\right]^{-1}, P_\text{max} \right\rbrace.
    \label{eq:gen_alg}
\end{align}
In the following we show that \eqref{eq:gen_alg} converges to a fixed point by using the standard interference function framework by Yates \cite{yates1995framework}. Furthermore, it is shown that the right-hand side in \eqref{eq:gen_alg} is a standard interference function for an arbitrary standard and feasible interference function $\mathcal{I}(\cdot)$, if $\mathcal{I}(\cdot)$ is $\log$-concave and the gradient w.r.t. to $\bf p$ is scale invariant as a function in $\bf q$. 
\begin{theorem}\label{thm:main}
    Let $\mathcal{I}$ be a standard and feasible interference function. \eqref{eq:gen_alg} converges to a fixed-point if $\mathcal{I}_i(\cdot)$ is $\log$-concave for  all $i = 1,\dots, K$ and the gradient of $\mathcal{I}_i$ are scale invariant, i.e. $\nabla_{\bf p} \mathcal{I}_i(\alpha {\bf q}) = \nabla_{\bf p} \mathcal{I}_i({\bf q})$ for every $\alpha, \bf q$ and $i = 1,\dots, K$. Moreover, the right-hand side in \eqref{eq:gen_alg} is a standard and feasible Interference Function.
\end{theorem}
\begin{proof}
See appendix.
\end{proof}

To justify the assumptions in the latter result, we investigate two examples. The following functions are meeting the assumptions in Theorem \ref{thm:main}.\\
\textbf{Examples:}
\begin{itemize}
    \item[a)] Consider the affine linear interference function, given in \eqref{eq:aff_int}. As the logarithm of an affine linear function is also $\log$-concave and the gradient of $\mathcal{I}_i$ only consists of the channel gains $\bf G_{ij}$ both conditions are fulfilled. Moreover the algorithm derived by using \eqref{eq:aff_int} in \eqref{eq:gen_alg} was also derived in \cite{tan2012fast, dahrouj2012power}. Here, convergence was also shown by using Yates interference function framework. 
\item[b)]It is known \cite{papandriopoulos2005optimal}, that the following function is an interference function
\begin{align}
    \mathcal{I}_i({\bf p}) = \sigma_i + \sum_{j\neq i} {\bf p}_i\log(1 +\frac{{\bf G}_{ij}{\bf p}_j}{{\bf p}_i})\label{eq:rayleigh}.
\end{align}
This interference function also fulfills the conditions in Theorem \ref{thm:main}. First of all \eqref{eq:rayleigh} is a non-negative concave \cite{tan2015optimal} function and thus also $\log$-concave. The gradient of \eqref{eq:rayleigh} is scale invariant, since
\begin{align*}
    \left(\nabla_{\bf p}\mathcal{I}_i({\bf p})\right)_i & = \sum_{j\neq i}\log(1 +\frac{{\bf G}_{ij}{\bf p}_j}{{\bf p}_i}) - \frac{{\bf G}_{ij}{\bf p}_j}{{\bf G}_{ij}{\bf p}_j + {\bf p}_i} \\
    & = \sum_{j\neq i}\log(1 +\frac{{\bf G}_{ij}{\alpha\bf p}_j}{\alpha{\bf p}_i}) - \frac{{\bf G}_{ij}\alpha{\bf p}_j}{{\bf G}_{ij}\alpha{\bf p}_j +\alpha {\bf p}_i} \\&= \left(\nabla_{\bf p}\mathcal{I}_i(\alpha{\bf p})\right)_i . 
\end{align*}
For $k\neq i$ it holds that
\begin{align*}
    \left(\nabla_{\bf p}\mathcal{I}_i({\bf p})\right)_j &= \frac{{\bf G}_{ij}{\bf p}_i}{{\bf G}_{ij}{\bf p}_j + {\bf p}_i}\\&=\frac{{\bf G}_{ij}\alpha{\bf p}_i}{{\bf G}_{ij}\alpha{\bf p}_j +\alpha {\bf p}_i} \\& = \left(\nabla_{\bf p}\mathcal{I}_i(\alpha{\bf p})\right)_j.
\end{align*}
\end{itemize}

\subsection{Weighted Sum Rate Maximization}
In Section \ref{sec:special_case}, we derived a closed-form update rule that ensures convergence to a stationary point for the specific case of \eqref{eq:sr}. Building on this foundation, the subsequent section proposes a primal-dual algorithm based on the DCA methodology. By introducing an auxiliary variable $\mathbf{q}$ into the DCA sub-problems, a similar update structure for $\mathbf{p}$. This update is directly connected to the previous section.

Using again the properties of the logarithm, the objective can be rewritten as
\begin{align}
\sum_i^K w_i\log (\mathbf{G}_{ii}{\bf p }_i + \mathcal{I}_i({\bf p })) - \sum_i^K w_i \log\left(\mathcal{I}_i({\bf p })\right).
\end{align}
The associated DCA is given by
\begin{subequations}\label{eq:dc_for_sr}
\begin{align}
    {\bf p}^{(k + 1)}\in&\arg\max_{\bf p} \sum_i^K w_i\log (\mathbf{G}_{ii}{\bf p }_i + \mathcal{I}_i({\bf p })) \\& - \sum_i^K w_i \log\left(\mathcal{I}_i({\bf p }^{(k)})\right) \\& -\nabla_{{\bf p}} \left(\sum_i^K w_i\log\left(\mathcal{I}_i({\bf p}^{(k)})\right)\right)^T\left({\bf p} - {\bf p}^{(k)}\right).
\end{align}    
\end{subequations}

Moreover, under the assumption that the interference function is {$\log$-concave} the algorithm given in \eqref{eq:dc_for_sr} converges to a stationary point of \eqref{eq:sr}. 

Introduce $\bf q$ and the corresponding constraint as the following
\begin{subequations}\label{eq:PDA_formulation}
    \begin{align}
    \max_{{\bf p, q},\, \text{s.t. } \bf p = q} &\sum_i^K w_i\log (\mathbf{G}_{ii}{\bf p }_i + \mathcal{I}_i({\bf q })) \\&- \sum_i^K w_i \left( \log\left(\mathcal{I}_i({\bf p }^{(k)})\right)\right) \\& - \sum_i^K w_i\left(\nabla_{\bf p}\log\left(\mathcal{I}_i({\bf p }^{(k)})\right)^T({\bf p} - {\bf p}^{(k)})\right).
\end{align}
\end{subequations}
Problem \eqref{eq:PDA_formulation} is equivalent to \eqref{eq:dc_for_sr} and is convex in $\left({\bf p,q}\right)$. The Lagrangian is given as
\begin{subequations}\label{eq:lagrangian}
\begin{align}
    \mathcal{L}&\left({\bf p,q,\mathbf{\lambda}}\right) = \sum_i^K w_i\log (\mathbf{G}_{ii}{\bf p }_i + \mathcal{I}_i({\bf q }))   \\ &- \sum_i^K w_i \left( \log\left(\mathcal{I}_i({\bf p }^{(k)})\right) +\frac{\nabla_{\bf p}\mathcal{I}_i({\bf p }^{(k)})^T}{\mathcal{I}_i({\bf p }^{(k)})}({\bf p} - {\bf p}^{(k)})\right) \\ & - \boldsymbol{\lambda}^T({\bf p - q}) .
\end{align}
\end{subequations}
Thus the convex sub-problems in the DCA \eqref{eq:dc_for_sr} can be solved by the following Primal Dual Algorithm, 
\begin{subequations}
\begin{align}
    {\bf p}^{(l + 1)} &\in \arg\max
    \mathcal{L}\left({{\bf p},{\bf q}^{(l)},\mathbf{\lambda}^{(l)}}\right)\label{eq:update_p}\\
    {\bf q}^{(l + 1)} &\in \arg\max
    \mathcal{L}\left({{\bf p}^{(l+1)},{\bf q},\mathbf{\lambda}^{(l)}}\right)\label{eq:update_q}\\
    \boldsymbol{\lambda}^{(l+1)} &= \boldsymbol{\lambda}^{(l)} + \alpha\left({\bf p}^{(l+1)} -{\bf q}^{(l+1)}\right). \label{eq:update_lam}    
\end{align}
\end{subequations}

The advantage of the choice of $\bf q$ lies in the gradients, as it allows to derive a closed-form update rule for $\mathbf{p}$, which closely resembles a structure with which we are already familiar with, i.e.
\begin{align*}
    \nabla_{\bf p} \mathcal{L}\left({\bf p,q, \boldsymbol{\lambda}}\right) & = \sum_i^K w_i\frac{\mathbf{G}_{ii}{\bf e }_i }{\mathbf{G}_{ii}{\bf p }_i + \mathcal{I}_i({\bf q })} \\&-\sum_i^K w_i\frac{\nabla_{\bf p}\mathcal{I}_i({\bf p }^{(k)})}{\mathcal{I}_i({\bf p }^{(k)})} - \boldsymbol{\lambda} = {\bf 0}.
\end{align*}
It is again possible to isolate ${\bf p}$ as it only appears in the first term, analogously to the previous section, i.e. 
\begin{align}
     {\bf p}_i  &=w_i\left[\sum_j^K \frac{w_j\nabla_{\bf p}\mathcal{I}_j({\bf p }^{(k)})_i}{\mathcal{I}_j({\bf p }^{(k)})} + \boldsymbol{\lambda}_i \right]^{(-1)} - {\bf q}_i\gamma({\bf q})_i^{-1}\label{eq:sr_structure}
\end{align}
where $\gamma_i({\bf p}) = \frac{{\bf G}_{ii}{\bf p}_i}{\mathcal{I}_i({\bf p})}$. The right hand side is again a standard and feasible interference function under the assumptions in Theorem 1, if $\boldsymbol{\lambda},\bf q$ are constant.

Furthermore, as $\gamma \rightarrow \infty$, the objective in \eqref{eq:sr} can be approximated by \eqref{eq:WS_log}. Consequently, since $\gamma^{-1} \rightarrow 0$, \eqref{eq:sr_structure} converges to the fixed-point iteration described in Section \ref{sec:special_case}. 

The update \eqref{eq:update_p} can be replaced by \eqref{eq:sr_structure} and constitutes an Interference Function if $\boldsymbol{\lambda}$ and $\mathbf{q}$ are fixed, as proven in Theorem \ref{thm:main}. The problem in \eqref{eq:update_q} is concave, and \eqref{eq:update_lam} represents a gradient update.

Note that this approach considers two stages of iterations; the outer iterations based on the DCA and the inner iterations solve the convex sub-problems. Moreover, for each inner iteration another optimization problem has to be solved. Due to this computational complexity, a deep unfolding based Primal Dual Algorithm is proposed in the next Section.

\section{Deep Unfolding}
Deep unfolding interprets the iterations of an iterative algorithm as layers of a neural network whose parameters are optimized through training, using a stochastic gradient descent to minimize a loss function. The number of iterations is predetermined and are executed sequentially as stages of a multi-layer computation. During training, specific variables in each layer are replaced by trainable variables. Moreover, we propose to replace the update of $\bf q$, i.e. \eqref{eq:update_q}, by a Fully Connected Neural Network (FCNN). 
\subsection{Learned Primal Dual Algorithm}
Based on the previous discussion, a Learned Primal Dual Algorithm (LPDA) is proposed to approximate a solution to \eqref{eq:sr} for an arbitrary standard and feasible $\log$-concave interference function. See Algorithm \ref{alg:PDA}. By updating ${\bf q}^{(k + 1)}$ with the help of a FCNN $\Phi\left({\bf p}^{(k+1)}, {\bf G}; \Theta\right)$, we can skip the inner loop and solve the convex problem \eqref{eq:update_q}. The FCNN is designed as follows: The input $({\bf p, G})$ is reshaped to a dimension of $K(K+1)$ and used directly as input. The number of neurons in each layers decrease with the number of layers. The tanh activation function is used in the hidden layers, while sigmoid activation is used in the output layer, as the target vector $\bf q$ is scaled to the range $[0,1]$.
\begin{algorithm}
\textbf{Input: }$\mathbf{G}, \bf w$. Initialize ${\bf p}_i = 1, \forall i=1,\dots,K$ \\
 \For{$k = 0,\dots, N-1$}{
    ${\bf p}_i^{(k+1)} = \min\biggl\lbrace w_i\left[\sum_j^K \frac{w_j\nabla_p\mathcal{I}_j({\bf p }^{(k)})_i}{\mathcal{I}_j({\bf p }^{(k)})} + \boldsymbol{\lambda}_i ^{(k)}\right]^{(-1)} -  {\bf q}^{(k)}_i\gamma({\bf q}^{(k)})_i^{-1}, P_{\max}\biggr\rbrace$\\
    ${\bf q}^{(k+1)} = {\bf \Phi}\left({\bf p}^{(k+1)}, {\bf G}; \Theta\right)$\\
    $\boldsymbol{\lambda}^{(k+1)} = \boldsymbol{\lambda}^{(k)} + \alpha^{(k)}\left({\bf p}^{(k+1)} -{\bf q}^{(k+1)}\right)$
 }
\caption{Learned Primal Dual Algorithm}\label{alg:PDA}
\end{algorithm}
\vspace{-5mm}
\subsection{Training}
The trainable variables of the FCNN are given by $\Theta$. Furthermore, while training algorithm \ref{alg:PDA} the step-sizes $\alpha^{(k)}$ are added to the set of trainable parameters. The Primal Dual Algorithm is then trained unsupervised and end-to-end for a fixed number of iterations $N$. The loss function is given as the objective in \eqref{eq:sr}, i.e. $ \ell\left({\bf p}^{(N)}\left(\mathbf{G}, \tilde{\Theta}\right), \mathbf{w}\right) = -\sum_{i=1}^Kw_iR_i\left({\bf p}^{(N)}\left(\mathbf{G}, \tilde{\Theta}\right)\right),$ where ${\bf p}^{(N)}$ is the ouput of Algorithm \ref{alg:PDA} and $\tilde{\Theta} = \Theta \cup \bigcup_{i = 0}^{N-1}\lbrace\alpha^{(k)}\rbrace$ is the set of all trainable variables. Training is carried out with standard algorithms, specifically the ADAM optimizer \cite{kingma2014adam}, with a decreasing learning rate.
\section{Numerical Experiments}\label{sec:NumExp}
For the following numerical experiment, the distance-dependent path loss model outlined in ITU-1411 is considered. $K$ transmitters, denoted as $\text{Tx}_i$, are uniformly distributed within a square area of length $500$m. The corresponding receivers, $\text{Rx}_i$, are positioned within a disk centered around each transmitter, with a distance uniformly chosen between $d_\text{min}$ to $d_\text{max}$. The respective system parameters can be found in Table \ref{tab:system_params}. For simplicity the same transmit power for each transmitter is used, i.e. $P_i = P$ for all $i$. The Primal Dual Algorithm is trained with a fixed number of $N=8$ iterations, the FCNN has $7$ layers, the number of neurons per hidden layer are $\lbrace 154, 132, 110, 88, 66, 44\rbrace$. In each training iteration $n_\text{train}$-random D2D-networks are generated as well as random uniform weights, thus matrix $\bf G$ and weight vector $\bf w$ are computed. Interference Function \eqref{eq:aff_int} is used.
\begin{table}[]
    \centering
    \begin{tabular}{c|c}
        Bandwidth & $20$ MHz \\ \hline
        Carrier frequency & $2.4$ GHz\\ \hline
        Antenna height & $1.5$ m\\ \hline
        Transmit power level & $20$ dBm\\ \hline
        Background noise level & $-174$ dBm/Hz \\ \hline
        $d_\text{min} \text{ \raisebox{-0.9ex}{\~{}} } d_\text{max}$ & $2 \text{ m}\text{ \raisebox{-0.9ex}{\~{}} } 65$ m 
    \end{tabular}
    \caption{System Parameters for ITU-1411 short-range outdoor model used in the following numerical experiments.
    \vspace{-5mm}}
    \label{tab:system_params}
\end{table}

The performance of the proposed LPDA is investigated with respect to FPLinQ \cite{shen_fplinq_2017, shen_fractional_2018, shen_fractional_2018-1}. 
The following performance metric is used
\begin{align}
\mathbb{E}_{\mathbf{G}, \mathbf{w}}\left[\frac{\sum_{i=1}^K w_{i}R_i\left(\hat{{\bf p}}, \mathbf{G}\right)}{\sum_{i=1}^K w_{i}R_i\left(\hat{{\bf p}}_{\text{FP}}, \mathbf{G}\right) }\right], \label{eq:performance}
\end{align}
where $\hat{{\bf p}}$ is the output of the LPDA with $N$ iterations and $\hat{{\bf p}}_{\text{FP}}$ is the output of FPLinQ after $100$ iterations.
Results can be seen in Table \ref{tab:PDA_vanilla} and Figure \ref{fig:weighted_true_obj} for $500$ unseen D2D-networks, each with $K=10$ links. Table \ref{tab:PDA_vanilla} shows that the proposed LPDA, trained with random uniform weights, is able to achieve the same performance as FPLinQ. Performance is less than $100\%$ in the case of weights equal to one, since these were not seen during training. Moreover, Figure \ref{fig:weighted_true_obj} presents the mean WSR over iterations of the benchmark as well as LPDA. The figure shows faster convergence compared to FPLinQ.

\begin{table}[h!]
    \centering
    \begin{tabular}{c|c|c}
    $n_\text{train}$ & $w_i\text{ \raisebox{-0.9ex}{\~{}} }\mathcal{U}_{[0,1]}$& $\mathbf{w}=\mathbf{1}$ \\ \hline
    $500$ & $\mathbf{101.2}\%$ & $97.65\%$  \\ \hline
    $1000$ & $\mathbf{101.4}\%$ & $97.87\%$ 
    
\end{tabular}

    \caption{Performance \eqref{eq:performance} of LPDA with $N=8$ and trained with random uniform weights for each D2D network.
    \vspace{-5mm}}
    \label{tab:PDA_vanilla}
\end{table}

\begin{figure}
\centering
\includegraphics[width=.9\linewidth]{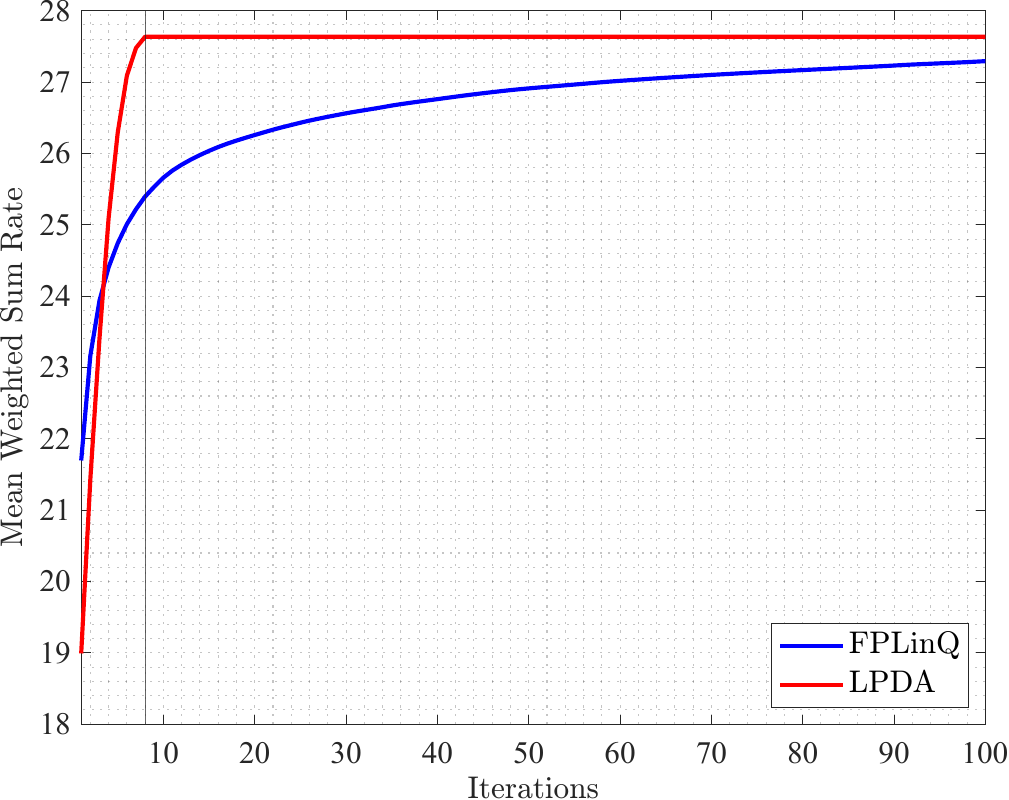}
\caption{Mean WSR versus iterations of the benchmark FPLinQ and the trained proposed LPDA \ref{alg:PDA} is evaluated across $500$ unseen D2D-networks with $K=10$ users and random uniform weights. LPDA is trained with $n_\text{train}=1000$, using random uniform weights and $N=8$ iterations, as indicated by the vertical black line in the figure.
\vspace{-5mm}}
\label{fig:weighted_true_obj}
\end{figure}
\section{Conclusion}
By combining the interference function framework by Yates with the difference-of-convex function programming framework, we developed a novel approach to address the WSR maximization problem for arbitrary interference functions. Specifically, we established theoretical guarantees for our approach when the respective interference function is $\log$-concave. Furthermore, we discovered that in a particular instance of the WSR problem, the resulting algorithm is a standard power control algorithm as defined by Yates. Finally, we proposed a learned PDA algorithm that, through a particular choice of auxiliary variable, reveals a connection to the previously described algorithm. In numerical examples, the competitiveness w.r.t. the benchmark of the proposed learned algorithm is demonstrated.

{\footnotesize
  \bibliographystyle{abbrv}
  \bibliography{bib}

\begin{thebibliography}{10}

\bibitem{al2011achieving}
H.~Al-Shatri and T.~Weber.
\newblock Achieving the maximum sum rate using dc programming in cellular
  networks.
\newblock {\em IEEE Transactions on signal processing}, 60(3):1331--1341, 2011.

\bibitem{cui_spatial_2019}
W.~Cui, K.~Shen, and W.~Yu.
\newblock Spatial {Deep} {Learning} for {Wireless} {Scheduling}.
\newblock {\em IEEE Journal on Selected Areas in Communications},
  37(6):1248--1261, June 2019.
\newblock arXiv: 1808.01486.

\bibitem{dahrouj2012power}
H.~Dahrouj, W.~Yu, and T.~Tang.
\newblock Power spectrum optimization for interference mitigation via iterative
  function evaluation.
\newblock {\em EURASIP Journal on Wireless Communications and Networking},
  2012:1--14, 2012.

\bibitem{TWC2012}
H.~R. Feyzmahdavian, M.~Johansson, and T.~Charalambous.
\newblock Contractive interference functions and rates of convergence of
  distributed power control laws.
\newblock In {\em 2012 IEEE International Conference on Communications (ICC)},
  pages 4395--4399, Ottawa, ON, Canada, 2012.

\bibitem{geng_optimality_2015}
C.~Geng, N.~Naderializadeh, A.~S. Avestimehr, and S.~A. Jafar.
\newblock On the {Optimality} of {Treating} {Interference} as {Noise}.
\newblock {\em IEEE Transactions on Information Theory}, 61(4):1753--1767, Apr.
  2015.

\bibitem{gregor2010learning}
K.~Gregor and Y.~LeCun.
\newblock Learning fast approximations of sparse coding.
\newblock In {\em Proceedings of the 27th international conference on
  international conference on machine learning}, pages 399--406, 2010.

\bibitem{kingma2014adam}
D.~P. Kingma and J.~Ba.
\newblock Adam: A method for stochastic optimization.
\newblock {\em arXiv preprint arXiv:1412.6980}, 2014.

\bibitem{lanckriet2009convergence}
G.~Lanckriet and B.~K. Sriperumbudur.
\newblock On the convergence of the concave-convex procedure.
\newblock {\em Advances in neural information processing systems}, 22, 2009.

\bibitem{le2018dc}
H.~A. Le~Thi and T.~Pham~Dinh.
\newblock Dc programming and dca: thirty years of developments.
\newblock {\em Mathematical Programming}, 169(1):5--68, 2018.

\bibitem{lee_graph_2020}
M.~Lee, G.~Yu, and G.~Y. Li.
\newblock Graph {Embedding} based {Wireless} {Link} {Scheduling} with {Few}
  {Training} {Samples}.
\newblock {\em arXiv:1906.02871 [cs, eess]}, Nov. 2020.
\newblock arXiv: 1906.02871.

\bibitem{lee2018deep}
W.~Lee, M.~Kim, and D.-H. Cho.
\newblock Deep learning based transmit power control in underlaid
  device-to-device communication.
\newblock {\em IEEE Systems Journal}, 13(3):2551--2554, 2018.

\bibitem{li_graph-based_2022}
B.~Li, G.~Verma, and S.~Segarra.
\newblock Graph-based {Algorithm} {Unfolding} for {Energy}-aware {Power}
  {Allocation} in {Wireless} {Networks}.
\newblock {\em arXiv:2201.11799 [cs, eess]}, Jan. 2022.
\newblock arXiv: 2201.11799.

\bibitem{boydccp1}
T.~Lipp and S.~Boyd.
\newblock Variations and extension of the convex--concave procedure.
\newblock {\em Optim Eng}, 17:263--287, 2016.

\bibitem{naderializadeh_itlinq_2014}
N.~Naderializadeh and A.~S. Avestimehr.
\newblock {ITLinQ}: {A} {New} {Approach} for {Spectrum} {Sharing} in
  {Device}-to-{Device} {Communication} {Systems}.
\newblock {\em arXiv:1311.5527 [cs, math]}, June 2014.
\newblock arXiv: 1311.5527.

\bibitem{papandriopoulos2005optimal}
J.~Papandriopoulos, J.~Evans, and S.~Dey.
\newblock Optimal power control for rayleigh-faded multiuser systems with
  outage constraints.
\newblock {\em IEEE Transactions on Wireless Communications}, 4(6):2705--2715,
  2005.

\bibitem{shelim_geometric_2022}
R.~Shelim and A.~S. Ibrahim.
\newblock Geometric {Machine} {Learning} {Over} {Riemannian} {Manifolds} for
  {Wireless} {Link} {Scheduling}.
\newblock {\em IEEE Access}, 10:22854--22864, 2022.
\newblock Conference Name: IEEE Access.

\bibitem{shen_fplinq_2017}
K.~Shen and W.~Yu.
\newblock {FPLinQ}: {A} cooperative spectrum sharing strategy for
  device-to-device communications.
\newblock In {\em 2017 {IEEE} {International} {Symposium} on {Information}
  {Theory} ({ISIT})}, pages 2323--2327, Aachen, Germany, June 2017. IEEE.

\bibitem{shen_fractional_2018}
K.~Shen and W.~Yu.
\newblock Fractional {Programming} for {Communication} {Systems}—{Part} {I}:
  {Power} {Control} and {Beamforming}.
\newblock {\em IEEE Transactions on Signal Processing}, 66(10):2616--2630, May
  2018.

\bibitem{shen_fractional_2018-1}
K.~Shen and W.~Yu.
\newblock Fractional {Programming} for {Communication} {Systems}—{Part} {II}:
  {Uplink} {Scheduling} via {Matching}.
\newblock {\em IEEE Transactions on Signal Processing}, 66(10):2631--2644, May
  2018.

\bibitem{boydccp2}
X.~Shen, S.~Diamond, Y.~Gu, and S.~Boyd.
\newblock Disciplined convex-concave programming.
\newblock In {\em Proceedings IEEE Conference on Decision and Control}, pages
  1009--1014, December 2016.

\bibitem{sun2016majorization}
Y.~Sun, P.~Babu, and D.~P. Palomar.
\newblock Majorization-minimization algorithms in signal processing,
  communications, and machine learning.
\newblock {\em IEEE Transactions on Signal Processing}, 65(3):794--816, 2016.

\bibitem{tan2015optimal}
C.~W. Tan.
\newblock Optimal power control in rayleigh-fading heterogeneous wireless
  networks.
\newblock {\em IEEE/ACM Transactions on Networking}, 24(2):940--953, 2015.

\bibitem{tan2012fast}
C.~W. Tan, M.~Chiang, and R.~Srikant.
\newblock Fast algorithms and performance bounds for sum rate maximization in
  wireless networks.
\newblock {\em IEEE/ACM Transactions on Networking}, 21(3):706--719, 2012.

\bibitem{tao1986algorithms}
P.~D. Tao and E.~B. Souad.
\newblock Algorithms for solving a class of nonconvex optimization problems.
  methods of subgradients.
\newblock In {\em North-Holland Mathematics Studies}, volume 129, pages
  249--271. Elsevier, 1986.

\bibitem{vucic2010dc}
N.~Vucic, S.~Shi, and M.~Schubert.
\newblock Dc programming approach for resource allocation in wireless networks.
\newblock In {\em 8th international symposium on modeling and optimization in
  mobile, ad hoc, and wireless networks}, pages 380--386. IEEE, 2010.

\bibitem{wu_flashlinq_nodate}
X.~Wu, S.~Tavildar, S.~Shakkottai, T.~Richardson, J.~Li, R.~Laroia, and
  A.~Jovicic.
\newblock Flashlinq: A synchronous distributed scheduler for peer-to-peer ad
  hoc networks.
\newblock {\em IEEE/ACM Transactions on Networking}, 21(4):1215--1228, 2013.

\bibitem{yao2023globally}
C.~Yao and X.~Jiang.
\newblock A globally convergent difference-of-convex algorithmic framework and
  application to log-determinant optimization problems.
\newblock {\em arXiv preprint arXiv:2306.02001}, 2023.

\bibitem{yates1995framework}
R.~D. Yates.
\newblock A framework for uplink power control in cellular radio systems.
\newblock {\em IEEE Journal on selected areas in communications},
  13(7):1341--1347, 1995.

\bibitem{yi_optimality_2015}
X.~Yi and G.~Caire.
\newblock Optimality of {Treating} {Interference} as {Noise}: {A}
  {Combinatorial} {Perspective}.
\newblock {\em arXiv:1504.00041 [cs, math]}, Sept. 2015.
\newblock arXiv: 1504.00041.

\bibitem{yuille2001concave}
A.~L. Yuille and A.~Rangarajan.
\newblock The concave-convex procedure (cccp).
\newblock {\em Advances in neural information processing systems}, 14, 2001.

\end{thebibliography}
}

\onecolumn

\section*{Appendix}
\subsection{Proof of Theorem 1}
As a first step to the main result a general property of $\log$-concave functions is stated.
\begin{lemma}\label{lemma:mon_dec}
    Let $f:\Omega \mapsto \mathbb{R}$ be a non-negative and $\log$-concave function, where $\Omega \subset \mathbb{R}_+^K$ a convex subset, then for ${\bf x} \geq {\bf x'}$
    \begin{align*}
        \frac{\nabla_{\bf x} f({\bf x}) }{f({\bf x})} \leq\frac{\nabla_{\bf x} f({\bf x'}) }{f({\bf x'})}.
    \end{align*}. 
\end{lemma}
\begin{proof}
    Since $f$ is non-negative and $\log$-concave $\nabla^2_{\bf x} \log(f({\bf x}))$ is negative semi-definite. Thus, we have the follwing inequality
    \begin{align*}
        \left(\frac{\nabla_{\bf x} f({\bf x}) }{f({\bf x})} - \frac{\nabla_{\bf x} f({\bf x'}) }{f({\bf x'})}\right)^T\left({\bf x}- {\bf x'}\right)\leq 0\, \forall {\bf x, x'}\in \Omega.
    \end{align*}
    Thus especially for ${\bf x} \geq {\bf x'}$
    \begin{align*}
        \frac{\nabla_{\bf x} f({\bf x}) }{f({\bf x})} \leq\frac{\nabla_{\bf x} f({\bf x'}) }{f({\bf x'})}.
    \end{align*}
\end{proof}

With this, the main result can be proven.
\begin{proof}
    We show that the following function
    \begin{align}
        \tilde{\mathcal{I}}_i({\bf p}) = w_i \left[\sum_j^K\frac{w_j\nabla \mathcal{I}_j({\bf p})_i}{\mathcal{I}_j({\bf p})}\right]^{-1} \label{eq:int_iog}
    \end{align}
    is a standard and feasible interference function. \\
    \textbf{Feasibility}:\\
    Since we assume, that $\mathcal{I}_j(\cdot)$ is $\log$-concave for every $i$ it is known from the CCCP framework, that \eqref{eq:grads} converges to an optimum or saddle-point of \eqref{eq:WS_log}, see Theorem 2 \cite{yuille2001concave}. This means there exists some $\tilde{\bf p}$ s.t. 
    \begin{align*}
        \tilde{\bf p}_i=\min\left\lbrace \tilde{\mathcal{I}}_i(\tilde{\bf p}) , P_\text{max} \right\rbrace
    \end{align*}
    and therefore
    \begin{align*}
        \tilde{\bf p}\geq \tilde{\mathcal{I}}(\tilde{\bf p}).
    \end{align*}\\
    \noindent\textbf{Standard:}\\
    \textbf{Positivity:} Since $\mathcal{I}({\bf p}) $ is assumed to be a standard interference function it holds that $\nabla \mathcal{I}_j({\bf p})  >0$ and $\mathcal{I}({\bf p}) > 0$, therefore $\tilde{\mathcal{I}}({\bf p}) > 0$.\\
    \textbf{Monotonicity:} Since $\log(\mathcal{I}_j)$ is $\log$-concave we know, with Lemma \ref{lemma:mon_dec}, for $\bf p\geq p'$ 
        \begin{align}
            \frac{w_j\nabla \mathcal{I}_j({\bf p})_i}{\mathcal{I}_j({\bf p})} & \leq \frac{w_j\nabla \mathcal{I}_j({\bf p}')_i}{\mathcal{I}_j({\bf p}')}\\
            \iff \sum_j^K\frac{w_j\nabla \mathcal{I}_j({\bf p})_i}{\mathcal{I}_j({\bf p})} &\leq \sum_j^K\frac{w_j\nabla \mathcal{I}_j({\bf p}')_i}{\mathcal{I}_j({\bf p}')} \\
            \iff w_i\left[\sum_j^K\frac{w_j\nabla \mathcal{I}_j({\bf p})_i}{\mathcal{I}_j({\bf p})}\right]^{-1} &\geq w_i\left[\sum_j^K\frac{w_j\nabla \mathcal{I}_j({\bf p}')_i}{\mathcal{I}_j({\bf p}')} \right]^{-1}\\
            \iff \tilde{\mathcal{I}}_i({\bf p}) & \geq \tilde{\mathcal{I}}_i({\bf p}') 
        \end{align}
        \textbf{Scalability:} Let $\alpha > 1$. Consider the following inequality, since the gradients are scale-invariant and $\mathcal{I}$ is a standard interference function we have,
        \begin{align*}
            \frac{1}{\alpha}\frac{\nabla_{\bf p} \mathcal{I}_j({\bf p})_i \mathcal{I}_j(\alpha{\bf p})}{\nabla_{\bf p} \mathcal{I}_j(\alpha{\bf p})_i\mathcal{I}_j({\bf p})} = \frac{1}{\alpha}\frac{\mathcal{I}_j(\alpha{\bf p})}{\mathcal{I}_j({\bf p})} < 1.
        \end{align*}
        Therefore
        \begin{align*}
            \frac{1}{\alpha}\frac{\nabla_{\bf p} \mathcal{I}_j({\bf p})_i}{\mathcal{I}_j({\bf p})} &< \frac{\nabla_{\bf p} \mathcal{I}_j(\alpha{\bf p})_i}{\mathcal{I}_j(\alpha{\bf p})} \\
            \iff \frac{1}{\alpha}\sum_j^K\frac{w_j\nabla_{\bf p} \mathcal{I}_j({\bf p})_i}{\mathcal{I}_j({\bf p})} &< \sum_j^K\frac{w_j\nabla_{\bf p} \mathcal{I}_j(\alpha{\bf p})_i}{\mathcal{I}_j(\alpha{\bf p})} \\
            \iff \alpha w_i\left[\sum_j^K\frac{\nabla_{\bf p} \mathcal{I}_j({\bf p})_i}{\mathcal{I}_j({\bf p})}\right]^{-1} &> w_i\left[\sum_j^K\frac{\nabla_{\bf p} \mathcal{I}_j(\alpha{\bf p})_i}{\mathcal{I}_j(\alpha{\bf p})} \right]^{-1}\\
            \iff \alpha \tilde{\mathcal{I}}_i({\bf p}) &> \tilde{\mathcal{I}}_i(\alpha{\bf p})
        \end{align*}
    Therefore $\tilde{\mathcal{I}}$ is a standard interference function and thus $\min\lbrace\tilde{\mathcal{I}}, P_\text{max}\rbrace$. The statement follows then by Theorem 2 and 4 in \cite{yates1995framework}.
\end{proof}
\end{document}